\documentclass[conference,a4paper]{IEEEtran}

\usepackage[utf8]{inputenc} 
\usepackage[T1]{fontenc}
\usepackage{url}
\usepackage{ifthen}
\usepackage{cite}
\usepackage[cmex10]{amsmath} 

\usepackage{amssymb, amsthm, amsfonts, mathtools}

\newtheorem{theorem}{Theorem}
\newtheorem{lemma}{Lemma}
\newtheorem{corollary}{Corollary}

\theoremstyle{plain}
\newtheorem{example}{Example}

\usepackage{graphicx}
\usepackage{subcaption}

\newcommand{\part}[2]{\frac{\partial #1}{\partial #2}}
\newcommand{\N}{\mathcal{N}}
\newcommand{\R}{\mathbb{R}}
\newcommand{\E}{\mathbb{E}}
\DeclareMathOperator{\sech}{sech}
\newcommand{\HS}{\mathrm{HS}}
\DeclareMathOperator{\Tr}{Tr}
\DeclareMathOperator{\Var}{Var}
\DeclareMathOperator{\Cov}{Cov}

\begin{document}
\title{Convexity of mutual information\\along the heat flow} 

\author{\IEEEauthorblockN{Andre Wibisono \,and\, Varun Jog}
\IEEEauthorblockA{Department of Electrical \& Computer Engineering \\
University of Wisconsin - Madison\\
Madison, WI 53706 \\
Email: aywibisono@wisc.edu, vjog@wisc.edu}
}

\maketitle

\begin{abstract}
We study the convexity of mutual information along the evolution of the heat equation.
We prove that if the initial distribution is log-concave, then mutual information is always a convex function of time.
We also prove that if the initial distribution is either bounded, or has finite fourth moment and Fisher information, then mutual information is eventually convex, i.e., convex for all large time.
Finally, we provide counterexamples to show that mutual information can be nonconvex at small time.
\end{abstract}

\section{Introduction}

The heat equation plays a fundamental role in many fields.
In thermodynamics, it describes the diffusion of heat in a body due to temperature differences.
In probability theory, it describes the evolution of the Brownian motion.
In information theory, it describes the additive white Gaussian noise channel, which is one of the most important communication channels.
In general, the heat equation can be used to model the transport of any quantity in a medium via a diffusion process. 
It also forms the basis for more general stochastic processes, such as the Ornstein-Uhlenbeck process or the Fokker-Planck process.
Therefore, the heat equation has found applications in diverse scientific disciplines---from explaining the evolution of zebra stripes~\cite{Tur52} to modeling stock prices via the Black-Scholes formula~\cite{BlaSch73}.
We are interested in the heat flow, which is the flow of the heat equation in the space of random variables.

The properties of the heat flow are closely linked to entropy.
Indeed, one important interpretation of the heat flow is as the flow that increases entropy as fast as possible.
More precisely, heat flow is the gradient flow (i.e., the steepest descent flow) of negative entropy in the space of probability distributions with the Wasserstein metric structure~\cite{JKO98}.
In this paper we will not need this result, but only use a certain key identity in our calculation.
Nevertheless, this relation suggests an intricate connection between entropy and the heat flow.

The behavior of entropy along the heat flow has been long studied.
The gradient flow interpretation above shows that entropy is increasing along the heat flow.
In particular, De Bruijn's identity~\cite{Sta59} states that the time derivative of entropy along the heat flow is given by the Fisher information, which is always positive.
Moreover, entropy is a concave function of time along the heat flow.
This is because
the second time derivative of entropy along the heat flow is the negative of the second-order Fisher information~\cite{McKean66,Tos99,Vil00};
the latter identity also implies the concavity of entropy power along the heat flow~\cite{Cos85,Dembo89,Dembo91}.
It is further conjectured that the higher derivatives of entropy along the heat flow have alternating signs~\cite{McKean66,Vil02,Che15}.
In one dimension, this has been verified up to the fourth derivative~\cite{Che15};
in multi dimension, this is true for the third derivative when the initial distribution is log-concave~\cite{Tos15}.

On the other hand, 
the behavior of mutual information along the heat flow has been less explored.
Clearly mutual information is decreasing along the heat flow by the data processing inequality, since the heat flow is a Markov chain.
De Bruijn's identity implies that the time derivative of mutual information along the heat flow is the negative of the mutual Fisher information;
the latter is proportional to the minimum mean square error (mmse) of estimating the initial from the final distribution, thus recovering the I-MMSE relation for the additive Gaussian channel~\cite{GuoEtAl05}.
Similarly, the second time derivative of mutual information along the heat flow is the mutual version of the second-order Fisher information; unfortunately, it does not always have a definite sign.

In this paper we study the convexity of mutual information along the heat flow.
This amounts to determining when the mutual second-order Fisher information is positive along the heat flow.
We show that in general, the mutual second-order Fisher information is positive whenever the final distribution is log-concave.
Since the heat flow preserves log-concavity, this implies our first main result:
If the initial distribution is log-concave, then mutual information is always convex along the heat flow.
In some cases, for example when the initial distribution is bounded, the heat flow implies eventual log-concavity, which means the final distribution eventually becomes log-concave; this implies mutual information is eventually convex along the heat flow for these cases.
Furthermore, we prove that in general, regardless of log-concavity, mutual information is eventually convex along the heat flow whenever the initial distribution has finite fourth moment and Fisher information.

Unlike entropy, however, we show that mutual information can be nonconvex along the heat flow.
We provide explicit counterexamples, namely mixtures of point masses and mixtures of Gaussians, for which mutual information along the heat flow is nonconvex at small time; furthermore, by scaling we can arrange the region of nonconvexity to engulf any finite time.
We elaborate on these results below.

\section{Background and problem setup}

\subsection{The heat flow}

The heat equation in $\R^n$ is the partial differential equation:
$$\part{\rho}{t} = \frac{1}{2} \Delta \rho$$
where $\rho = \rho(x,t)$ for $x \in \R^n$, $t \ge 0$, and $\Delta = \sum_{i=1}^n \part{^2}{x_i^2}$ is the Laplacian operator.
This equation conserves mass, so if $\rho_0 = \rho(\cdot,0)$ is a probability distribution, then so is $\rho_t = \rho(\cdot,t)$ for all $t > 0$. 
The heat equation admits a closed-form solution via convolution:
$$\rho_t = \rho_0 \ast \gamma_t$$
where $\gamma_t(x) = (2\pi t)^{-\frac{n}{2}} e^{-\frac{\|x\|^2}{2t}}$ is the heat kernel at time $t$.
Probabilistically, if $X_0 \sim \rho_0$ is a random variable in $\R^n$, then $X_t \sim \rho_t$ that evolves following the heat equation is given by
$$X_t = X_0 + \sqrt{t} Z$$
where $Z \sim \N(0,I)$ is the standard Gaussian random variable in $\R^n$ independent of $X_0$.
We call this the heat flow.
(Note that the true solution to the heat equation is the Brownian motion, but at each time $t$ it has the same distribution as $X_t$ above.)
Observe that even when $X_0 \sim \rho_0$ has a singular density, 
$X_t \sim \rho_t$ has a smooth positive density for all $t > 0$.

\begin{example}
If $X_0 \sim \delta_{a}$ is a point mass at some $a \in \R^n$, then $X_t \sim \N(a, tI)$ is Gaussian with mean $a$ and covariance $tI$.
\end{example}

\begin{example}
If $X_0 \sim \N(\mu,\Sigma)$ is Gaussian, then $X_t \sim \N(\mu,\Sigma+tI)$ is also Gaussian with the same mean and increasing covariance.
\end{example}

\begin{example}
If $X_0 \sim \sum_{i=1}^k p_i \delta_{a_i}$ is a mixture of point masses, then $X_t \sim \sum_{i=1}^k p_i \N(a_i,tI)$ is a mixture of Gaussians with the same covariance $tI$.
\end{example}

\begin{example}
If $X_0 \sim \sum_{i=1}^k p_i \N(a_i, \Sigma_i)$ is a mixture of Gaussians, then $X_t \sim \sum_{i=1}^k p_i \N(a_i,\Sigma_i+tI)$ is also a mixture of Gaussians with the same means and increasing covariance.
\end{example}

\subsection{Entropy and Fisher information}

Let $X$ be a random variable in $\R^n$ with a smooth positive density $\rho$.

The (differential) {\em entropy} of $X \sim \rho$ is
$$H(X) = -\int_{\R^n} \rho(x) \log \rho(x) \, dx.$$

The {\em Fisher information} of $X \sim \rho$ is
$$J(X) = \int_{\R^n} \rho(x) \|\nabla \log \rho(x)\|^2 \, dx.$$

The {\em second-order Fisher information} of $X \sim \rho$ is
$$K(X) = \int_{\R^n} \rho(x) \|\nabla^2 \log \rho(x)\|_{\HS}^2 \, dx.$$
Here $\|A\|_{\HS}^2 = \sum_{i,j=1}^n A_{ij}^2 = \sum_{i=1}^n \lambda_i(A)^2$ is the Hilbert-Schmidt (or Frobenius) norm of a symmetric matrix $A = (A_{ij}) \in \R^{n \times n}$ with eigenvalues $\lambda_i(A) \in \R$.

In general we have the inequality
\begin{align}\label{Eq:KJ}
K(X) \ge \frac{J(X)^2}{n}
\end{align}
which is equivalent to the entropy power inequality~\cite{Cos85,Dembo89,Dembo91,Vil00}.

\begin{example}
If $X \sim \N(\mu,\Sigma)$ is Gaussian, then
\begin{align*}
H(X) &= \frac{1}{2} \log \det (2 \pi e \Sigma) = \frac{1}{2} \sum_{i=1}^n \log (2 \pi e \lambda_i) \\ 
J(X) &= \Tr(\Sigma^{-1}) = \sum_{i=1}^n \frac{1}{\lambda_i} \\
K(X) &= \|\Sigma^{-1}\|^2_{\HS} = \sum_{i=1}^n \frac{1}{\lambda_i^2}
\end{align*}
where $\lambda_1,\dots,\lambda_n > 0$ are the eigenvalues of $\Sigma \succ 0$.
\end{example}

Our interest in the first and second-order Fisher information is because they are the first and second derivatives of entropy along the heat flow.

\begin{lemma}\label{Lem:DerEnt}
Along the heat flow $X_t = X_0 + \sqrt{t} Z$,
\begin{align*}
\frac{d}{dt} H(X_t) &= \frac{1}{2} J(X_t) \\
\frac{d^2}{dt^2} H(X_t) &= -\frac{1}{2} K(X_t).
\end{align*}
\end{lemma}

Note that since $J(X_t) \ge 0$, the first derivative of entropy is positive, which means entropy is increasing along the heat flow.
Similarly, since $K(X_t) \ge 0$, the second derivative of entropy is negative, which means entropy is a concave function along the heat flow.

\subsection{Mutual information and mutual Fisher information}
\label{Sec:Mut}

Let $(X,Y)$ be a joint random variable in $\R^n \times \R^n$ with a joint density $\rho_{XY}$, which we can factorize into a product of marginal and conditional densities:
$$\rho_{XY}(x,y) = \rho_X(x) \, \rho_{Y|X}(y\,|\,x) = \rho_Y(y) \, \rho_{X|Y}(x\,|\,y).$$
We assume $\rho_Y$ and $\rho_{Y|X}(\cdot\,|\,x)$ are smooth and positive for all $x \in \R^n$.

The {\em mutual information} of $(X,Y)$ is
$$I(X;Y) = H(Y) - H(Y\,|\,X)$$
where $H(Y\,|\,X) = \int \rho_X(x) H(\rho_{Y|X}(\cdot\,|\,x))\,dx$ is the expected entropy of the conditional densities.

The {\em mutual Fisher information} of $(X,Y)$ is
$$J(X;Y) = J(Y\,|\,X) - J(Y)$$
where $J(Y\,|\,X) = \int \rho_X(x) J(\rho_{Y|X}(\cdot\,|\,x))\,dx$ is the expected Fisher information of the conditional densities.

The {\em mutual second-order Fisher information} of $(X,Y)$ is
$$K(X;Y) = K(Y\,|\,X) - K(Y)$$
where $K(Y\,|\,X) = \int \rho_X(x) K(\rho_{Y|X}(\cdot\,|\,x))\,dx$ is the expected second-order Fisher information of the conditional densities.

Mutual information is symmetric: $I(X;Y) = I(Y;X)$.
However, mutual first and second-order Fisher information are not symmetric: in general, $J(X;Y) \neq J(Y;X)$ and $K(X;Y) \neq K(Y;X)$.

The mutual Fisher information $J(X;Y)$ can be shown to be equal to the {\em backward (statistical) Fisher information} $\Phi(X\,|\,Y)$, which is manifestly positive. 
The mutual second-order Fisher information $K(X;Y)$, on the other hand, is not always positive, but it can be represented in terms of the {\em backward (statistical) second-order Fisher information} $\Psi(X\,|\,Y)$; see Appendix~\ref{App:ProofKJMut} for detail.

Analogous to the basic (non-mutual) inequality~\eqref{Eq:KJ}, we have the following result. 
Recall that a smooth probability distribution $\rho$ in $\R^n$ is {\em $\alpha$-log-semiconcave} for some $\alpha \in \R$ if
$$-\nabla^2 \log \rho(x) \succeq \alpha I~~~~\forall \, x \in \R^n.$$
When $\alpha \ge 0$, we say $\rho$ is log-concave.

\begin{lemma}\label{Lem:KJMut}
If $Y \sim \rho_Y$ is $\alpha$-log-semiconcave for some $\alpha \in \R$, then
$$K(X;Y) \ge \frac{J(X;Y)^2}{n} + 2\alpha J(X;Y).$$
\end{lemma}

In particular, if $\rho_Y$ is log-concave, then $K(X;Y) \ge 0$.

\subsection{Mutual information along the heat flow}
\label{Sec:MutHeat}

Now consider when $Y = X_t$ is the heat flow from $X = X_0$.

By the linearity of the channel, the identities for the derivatives of entropy in Lemma~\ref{Lem:DerEnt} imply the following identities for the derivatives of mutual information along the heat flow.

\begin{lemma}\label{Lem:DerMut}
Along the heat flow $X_t = X_0 + \sqrt{t} Z$,
\begin{align*}
\frac{d}{dt} I(X_0;X_t) &= -\frac{1}{2} J(X_0;X_t) \\
\frac{d^2}{dt^2} I(X_0;X_t) &= \frac{1}{2} K(X_0;X_t).
\end{align*}
\end{lemma}

Since $J(X_0;X_t) = \Phi(X_0\,|\,X_t) \ge 0$, the first identity above shows that mutual information is decreasing along the heat flow.
In fact along the heat flow $\Phi(X_0\,|\,X_t) = \frac{1}{t^2} \Var(X_0\,|\,X_t)$ is proportional to the mmse of estimating $X_0$ from $X_t$, thus recovering the I-MMSE relation for Gaussian channel~\cite{GuoEtAl05,WibisonoJL17}.
From the second identity above, we see that the convexity of mutual information along the heat flow is equivalent to the positivity of $K(X_0;X_t)$, for which Lemma~\ref{Lem:KJMut} will be useful.

Finally, we note that since $X_t \,|\, X_0$ is Gaussian, the various mutual quantities in Lemma~\ref{Lem:DerMut} are simply comparisons against a baseline Gaussian:
$I(X_0;X_t) = H(X_t) - \frac{n}{2} \log (2 \pi t e)$,
$$J(X_0;X_t) = \frac{n}{t} - J(X_t), ~~\text{ and }~~
K(X_0;X_t) = \frac{n}{t^2} - K(X_t).$$
In the opposite order, mutual information stays the same: $I(X_t;X_0) = I(X_0;X_t)$.
On the other hand,  
the mutual first and second-order Fisher information 
can be computed explicitly and do not depend on $X_t$:
$$J(X_t;X_0) = \frac{n}{t} ~~~~\text{ and }~~~~ K(X_t;X_0) = \frac{n}{t^2} + \frac{2}{t} J(X_0).$$
See Appendix~\ref{App:DetMutHeat} for detail.

\newpage
\section{Convexity of mutual information}

We present our main results on the convexity of mutual information along the heat flow.
Throughout, let $X_t = X_0 + \sqrt{t} Z$ denote the heat flow.

\subsection{Perpetual convexity when initial distribution is log-concave}

Recall from Lemma~\ref{Lem:KJMut} and~\ref{Lem:DerMut} that mutual information is convex whenever the final distribution is log-concave.
Since the heat flow preserves log-concavity, this implies mutual information is always convex when the initial distribution is log-concave.

\begin{theorem}\label{Thm:PerpConv}
If $X_0 \sim \rho_0$ has a log-concave distribution, then mutual information $t \mapsto I(X_0;X_t)$ is convex for all $t \ge 0$.
\end{theorem}

\subsection{Eventual convexity when initial distribution is bounded}

Next, we ask when the final distribution is eventually convex under the heat flow, which also implies the eventual convexity of mutual information.
We can show that if the initial distribution is bounded, then the final distribution is eventually log-concave; this fact has also been observed in~\cite{Lee03}.

We say a probability distribution $\rho$ is {\em $D$-bounded} for some $D \ge 0$ if it is supported on a domain of diameter at most $D$.

\begin{theorem}\label{Thm:EventConv}
If $X_0 \sim \rho_0$ has a $D$-bounded distribution, then mutual information $t \mapsto I(X_0;X_t)$ is convex for all $t \ge D^2$.
\end{theorem}

Since convolution with log-concave distribution preserves log-concavity, we also have the following corollary.
Note that when the bounded part is a point mass (with diameter $D = 0$) this recovers
Theorem~\ref{Thm:PerpConv} above.

\begin{corollary}\label{Cor:EventConv}
If $X_0 \sim \rho_0$ is a convolution of a $D$-bounded and a log-concave distribution, then mutual information $t \mapsto I(X_0;X_t)$ is convex for all $t \ge D^2$.
\end{corollary}

For example, if $X_0 \sim \sum_{i=1}^k p_i \N(a_i, \Sigma)$ is a mixture of Gaussians with the same covariance, then the bounded part $\sum_{i=1}^k p_i \delta_{a_i}$ has diameter $D = \max_{i \neq j} \|a_i-a_j\|$.

\subsection{Eventual convexity when Fisher information is finite}

We now investigate when mutual information is eventually convex in general, 
regardless of the log-concavity of the distributions.
We show that if the initial distribution has finite fourth moment and Fisher information, then mutual information is eventually convex.

For $p \ge 0$, let $M_p(X) = \E[\|X-\mu\|^p]$ denote the $p$-th moment of a random variable $X$ with mean $\E[X] = \mu \in \R^n$.

\begin{theorem}\label{Thm:EventConvFI}
If $X_0 \sim \rho_0$ has finite fourth moment $M_4(X_0) < \infty$ and Fisher information $J(X_0) < \infty$, then mutual information $t \mapsto I(X_0;X_t)$ is convex for all $t \ge \frac{1}{n^2} J(X_0)M_4(X_0)$.
\end{theorem}

Thus, we see that under a wide variety of conditions, mutual information is eventually convex along the heat flow. 
However, it turns out mutual information is {\em not} always convex along the heat flow, in contrast to the concavity of entropy or entropy power along the heat flow.

\section{Nonconvexity of mutual information}
\label{Sec:NonConv}

We present some counterexamples for which mutual information along the heat flow is not convex at some small time.
Concretely, we study mixtures of point masses and mixtures of Gaussians as initial distribution of the heat flow.

\subsection{Mixture of two point masses}
\label{Sec:MixtPoint}

Let 
$X_0 \sim \frac{1}{2} \delta_{-a} + \frac{1}{2} \delta_a$
be a uniform mixture of two point masses centered at $a$ and $-a$, for some $a \in \R^n$, $a \neq 0$.
Along the heat flow, 
$X_t \sim \frac{1}{2} \N(-a,tI) + \frac{1}{2} \N(a,tI)$
is a uniform mixture of two Gaussians with equal covariance $tI$.

For $u > 0$, let 
$$V_u = \N(u,u) \in \R$$
denote the one-dimensional Gaussian random variable with mean and variance both equal to $u$.
Then by direct calculation:
\begin{align*}
I(X_0;X_t) &= \frac{\|a\|^2}{t} - \E[\log \cosh(V_{\frac{\|a\|^2}{t}})] \\
J(X_0;X_t) &= \frac{\|a\|^2}{t^2} \E[\sech^2(V_{\frac{\|a\|^2}{t}})] \\
K(X_0;X_t) &= \frac{2\|a\|^2}{t^3} \E[\sech^2(V_{\frac{\|a\|^2}{t}})] \! - \! \frac{\|a\|^4}{t^4} \E[\sech^4(V_{\frac{\|a\|^2}{t}})].
\end{align*}
Note the dependence on dimension is only implicit via $\|a\|^2$.

The behavior of these quantities is illustrated in Figure~\ref{Fig:MixtPoint}.
Mutual information is not convex at small time since it starts at some finite value (in fact $\log 2$), and stays flat for a while before decreasing.
Its second derivative, the mutual second-order Fisher information, starts at $0$ and becomes negative before eventually becoming positive.
Thus, mutual information is concave for all small time.
Furthermore, by scaling $\|a\|^2$ we can stretch the region of nonconvexity to cover any finite time interval.

\subsection{Mixture of two Gaussians}
\label{Sec:MixtGaus}

Let $X_0 \sim \frac{1}{2} \N(-a,sI) + \frac{1}{2} \N(a,sI)$ be a uniform mixture of two Gaussians with the same covariance $sI$ for some $s > 0$, centered at $-a$ and $a$ for some $a \in \R^n$, $a \neq 0$.
Note, the limit $s \to 0$ recovers the mixture of two point masses above.
Along the heat flow, $X_t \sim \frac{1}{2} N(-a,(s+t)I) + \frac{1}{2} \N(a,(s+t)I)$ is also a mixture of two Gaussians with increasing covariance.

Then with $V_u = \N(u,u)$ as above, we have:
\begin{align*}
I(X_0;X_t) &= \frac{n}{2} \log\left(1+\frac{s}{t}\right) + \frac{\|a\|^2}{s+t} - \E[\log \cosh(V_{\frac{\|a\|^2}{s+t}})] \\
J(X_0;X_t) &= \frac{ns}{t(s+t)} + \frac{\|a\|^2}{(s+t)^2} \E[\sech^2(V_{\frac{\|a\|^2}{s+t}})] \\
K(X_0;X_t) &= \frac{ns(s+2t)}{t^2(s+t)^2} + \frac{2\|a\|^2}{(s+t)^3} \E[\sech^2(V_{\frac{\|a\|^2}{s+t}})] \\
&~~~~ - \frac{\|a\|^4}{(s+t)^4} \E[\sech^4(V_{\frac{\|a\|^2}{s+t}})].
\end{align*}
Note the explicit dependence on the dimension $n$.

The behavior of these quantities is illustrated in Figure~\ref{Fig:MixtGaus} for $n=1$.
Mutual information initially starts at $+\infty$, but it decreases quickly and exhibits a similar pattern of nonconvexity as the mixture of point masses.
Its second derivative, the mutual second-order Fisher information, also starts at $+\infty$, but decreases quickly and becomes negative for some time before eventually becoming positive.
Thus, mutual information is concave at some small time, and by scaling $\|a\|^2$ we can enlarge the region of nonconvexity.

\begin{figure}
    \centering
    \begin{subfigure}[b]{0.2311\textwidth}
        \includegraphics[width=\textwidth]{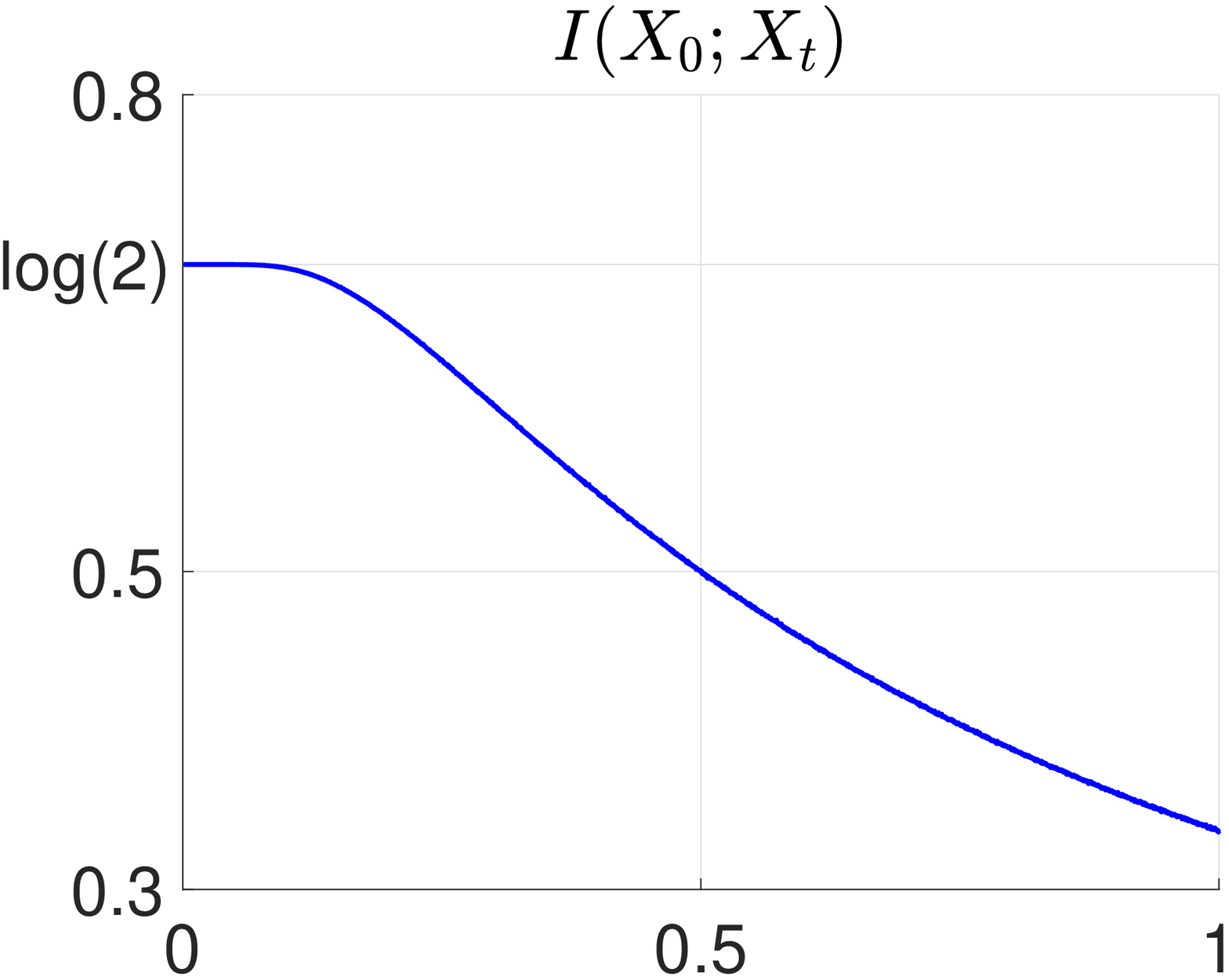}
    \end{subfigure}
      \; 
    \begin{subfigure}[b]{0.2311\textwidth}
        \includegraphics[width=\textwidth]{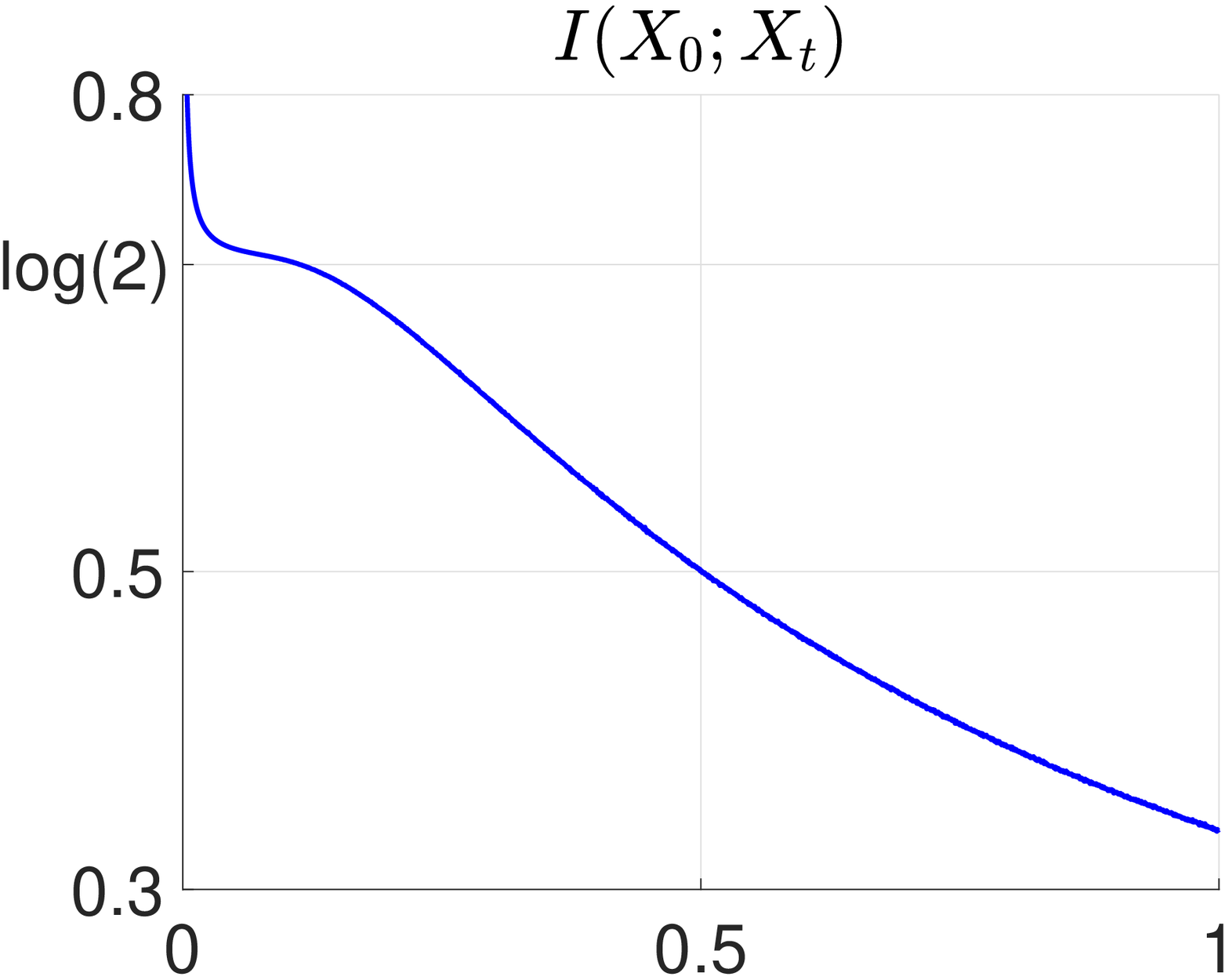}
    \end{subfigure} 

    \vspace{10pt}
    
    \begin{subfigure}[b]{0.22\textwidth}
        \includegraphics[width=\textwidth]{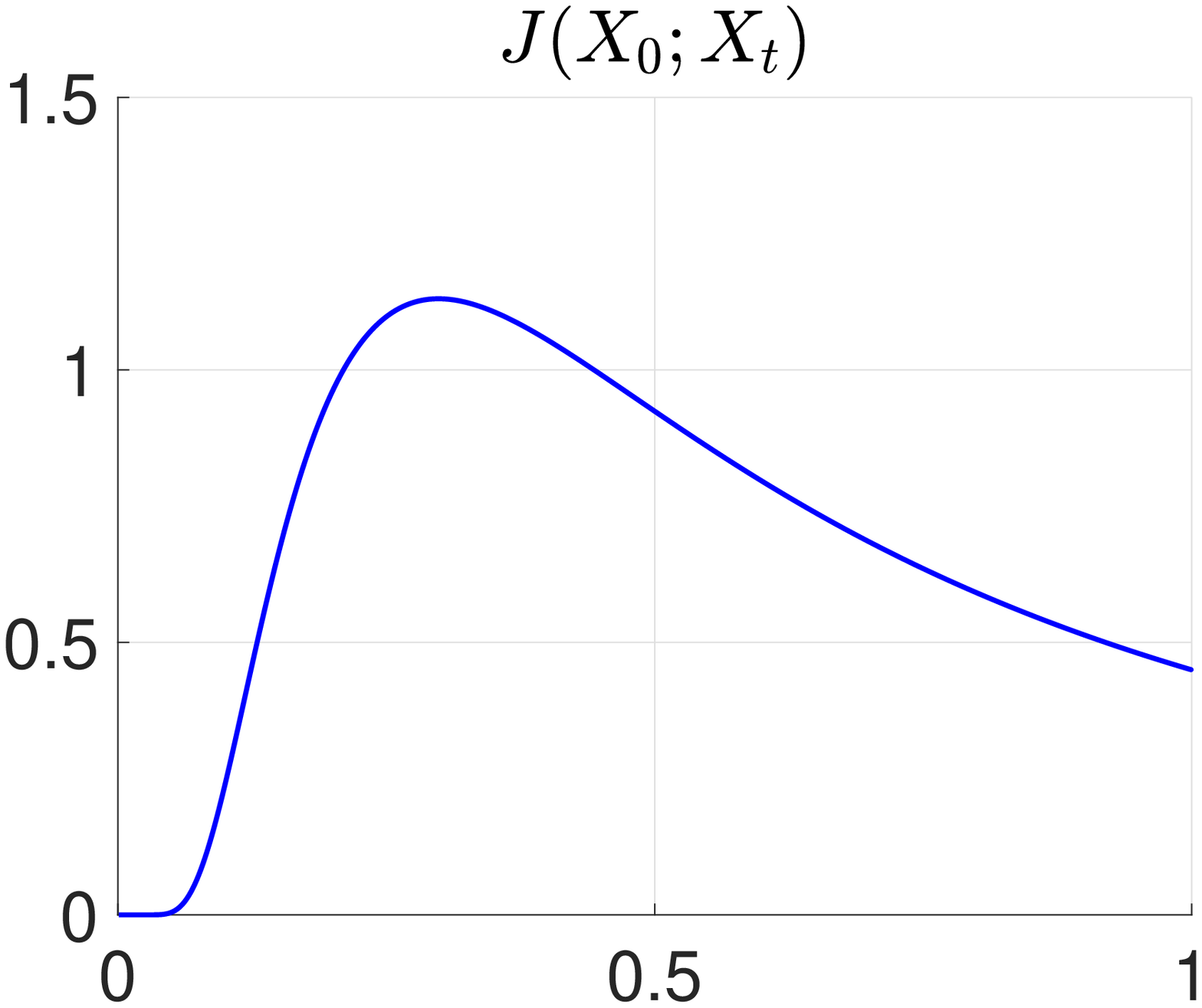}
    \end{subfigure}
     ~~ 
    \begin{subfigure}[b]{0.22\textwidth}
        \includegraphics[width=\textwidth]{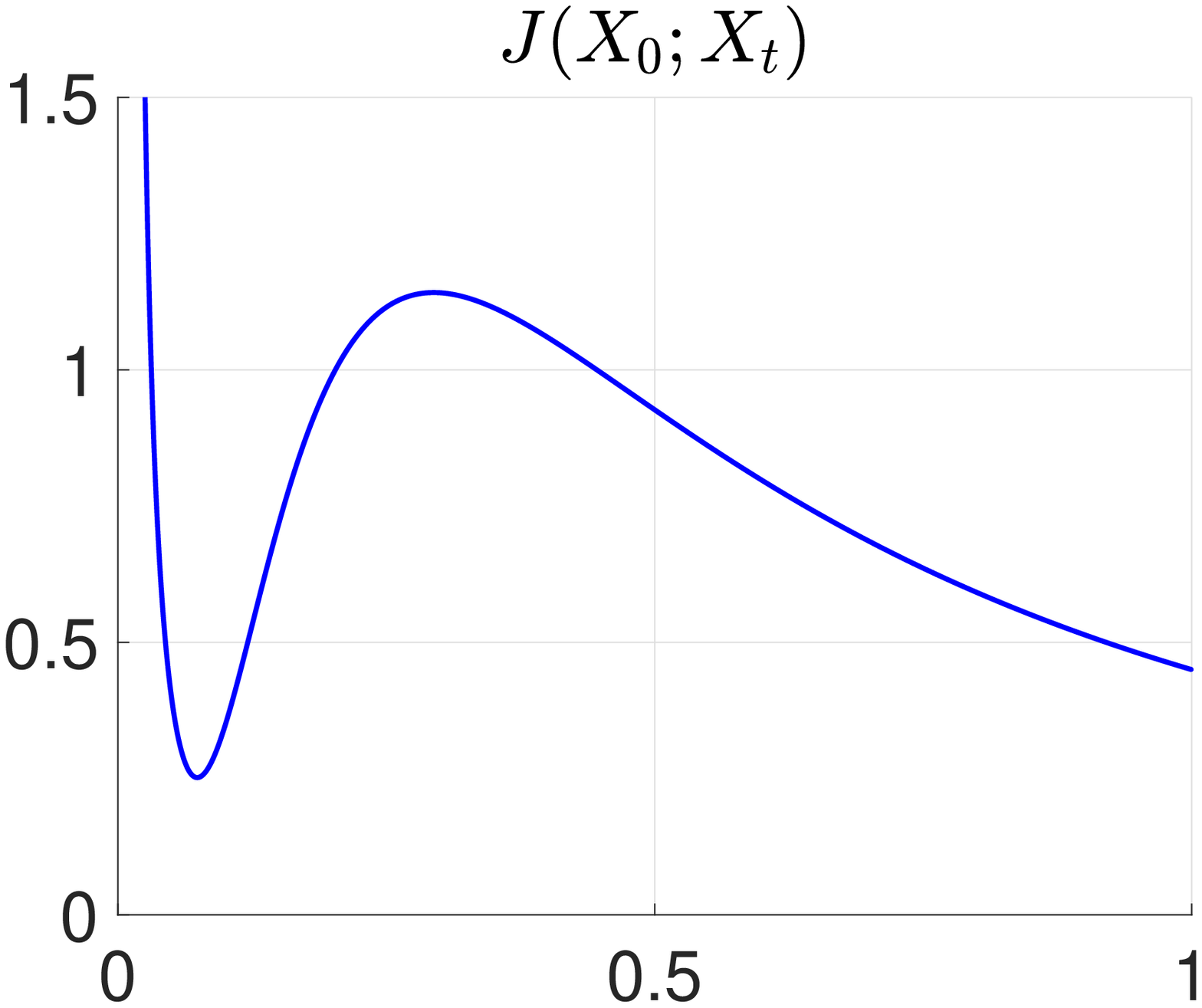}
    \end{subfigure}

    \vspace{10pt}
    
    \begin{subfigure}[b]{0.22\textwidth}
        \includegraphics[width=\textwidth]{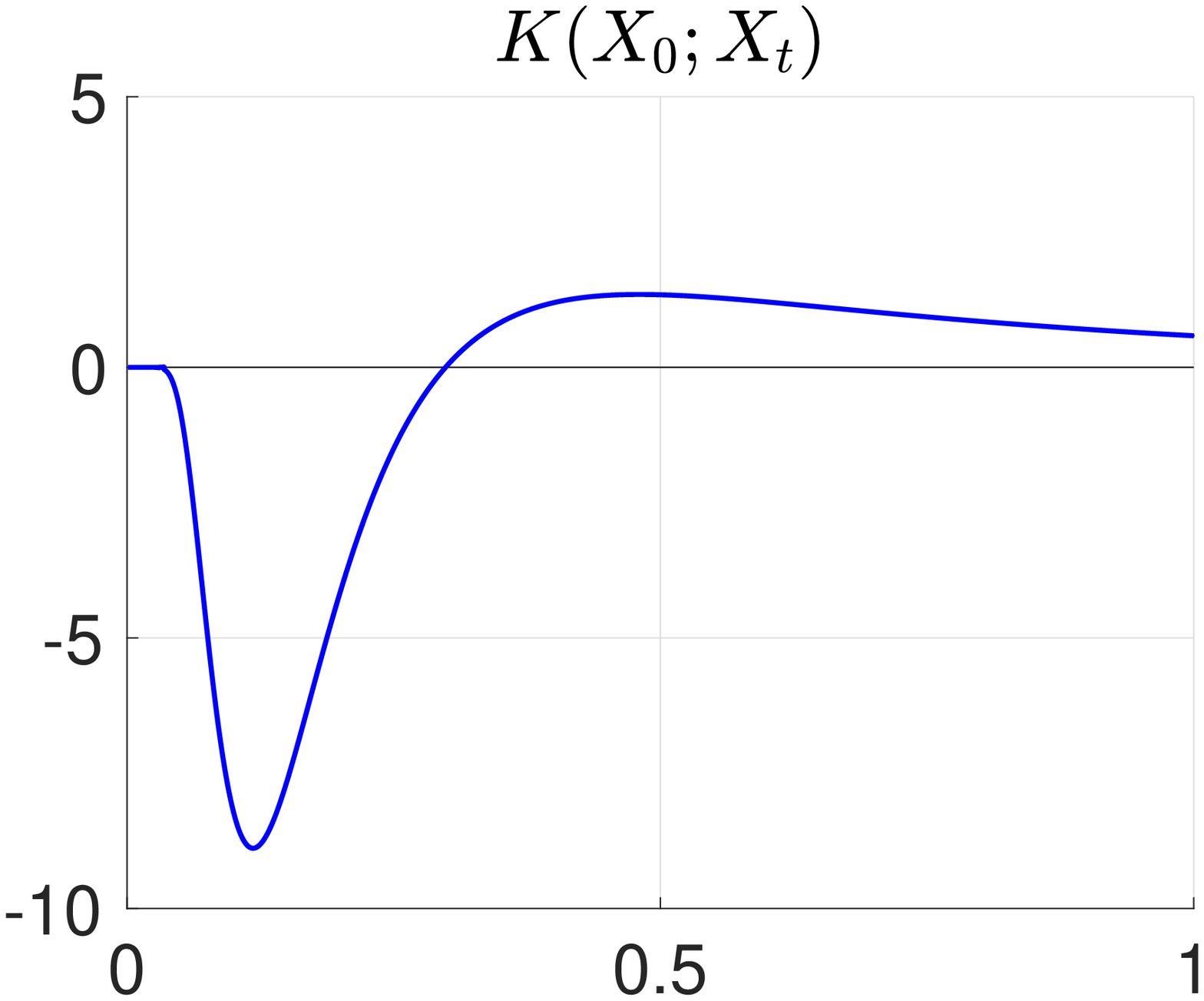}
        \caption{Mixture of point masses}
        \label{Fig:MixtPoint}
    \end{subfigure} 
      ~~ 
    \begin{subfigure}[b]{0.22\textwidth}
        \includegraphics[width=\textwidth]{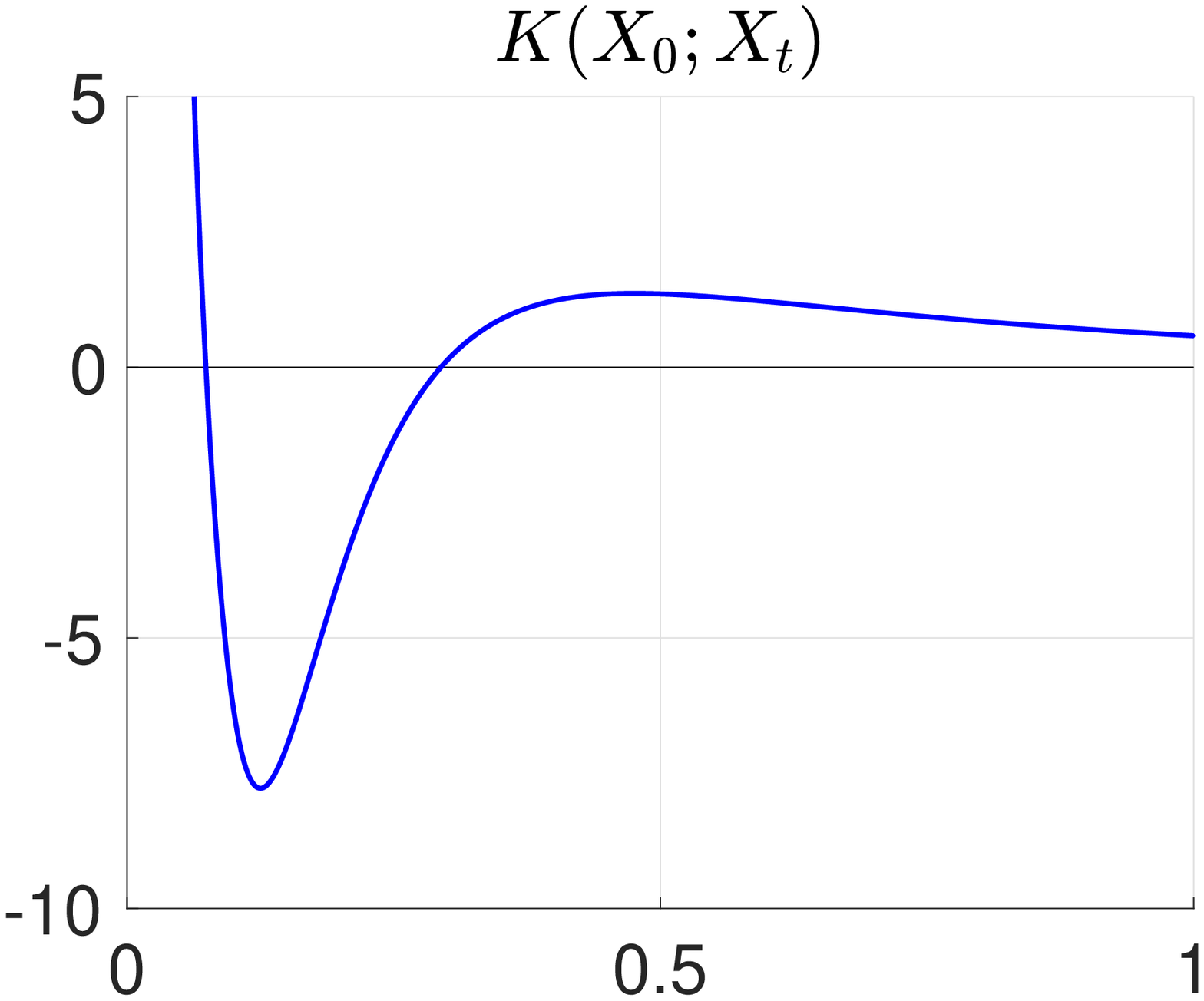}
        \caption{Mixture of Gaussians}
        \label{Fig:MixtGaus}
    \end{subfigure} 

    \caption{Behavior of mutual information and its two derivatives along the heat flow.
    (a) Left: $X_0 \sim \frac{1}{2} \delta_{-1} + \frac{1}{2} \delta_1$.
    (b) Right: $X_0 \sim \frac{1}{2} \N(-1,s) + \frac{1}{2} \N(1,s)$
     with $s = 10^{-3}$.}
    \label{Fig:Mixt}
\end{figure}

\subsection{General mixture of point masses}

Let $X_0 \sim \sum_{i=1}^k p_i \delta_{a_i}$ be a mixture of point masses centered at distinct $a_i \in \R^n$, with mixture probabilities $p_i > 0$, $\sum_{i=1}^k p_i = 1$.
Along the heat flow,
$X_t \sim \sum_{i=1}^k p_i \N(a_i,tI)$ is a mixture of Gaussians with increasing covariance $tI$ at the same centers.

We show that mutual information starts at a finite value which is equal to the discrete entropy of the mixture probability, and it is exponentially concentrated at small time.

Let $\|p\|_\infty = \max_{i,j} p_i/p_j$ and $m = \min_{i \neq j} \|a_i-a_j\| > 0$.
Let $h(p) = -\sum_{i=1}^k p_i \log p_i$ denote the discrete entropy.

\begin{theorem}\label{Thm:GenMixt}
For all $0 < t \le \frac{m^2}{676\|p\|_\infty^2}$,
$$0 \,\le\, h(p) - I(X_0;X_t) \,\le\, 3(k-1)\|p\|_\infty e^{-0.085 \frac{m^2}{t}}.$$
\end{theorem}

The theorem above implies that
$$\lim_{t \to 0} I(X_0;X_t) = h(p).$$
In particular, the initial value of mutual information does not depend on the locations of the centers, as long as they are distinct.
This is interesting, because by moving the centers and merging them we can obtain discontinuities of the mutual information with respect to the initial random variable at the origin
(moving the centers changes the mutual information curve but preserves the starting point, while merging the centers makes the starting point jump).

Furthermore, 
if a function converges exponentially fast, then all its derivatives must converge to zero exponentially fast.
Thus, we have the following corollary.

\begin{corollary}\label{Cor:Last}
For all $\ell \in \mathbb{N}$, $\lim_{t \to 0} \frac{d^\ell}{dt^\ell} I(X_0;X_t) = 0$. 
\end{corollary}

In particular, the first derivative of mutual information, which is negative mutual Fisher information, starts at $0$.
Since the initial distribution is bounded, mutual information is eventually convex by Theorem~\ref{Thm:EventConv}, which means mutual Fisher information is eventually decreasing. 
Since mutual Fisher information is always nonnegative, this means it must initially increase, during which mutual information is concave; this is similar to the behavior observed in~$\S\ref{Sec:MixtPoint}$.

Moreover,
by the continuity of the second-order Fisher information, this suggests that when the initial distribution is a mixture of Gaussians, mutual information may be also be concave at some small time.

\section{Discussion and future work}

In this paper we have studied the convexity of mutual information along the heat flow.
We have shown that under a wide variety of conditions mutual information is eventually convex, and we have shown examples where mutual information may be concave at some small time.
Many questions remain.

One question is how much we can extend the results to general stochastic processes.
We can show most of our results still hold for the Ornstein-Uhlenbeck process~\cite{WJOU18}.
For general Fokker-Planck processes the situation is more complicated, but at least there are explicit formulae for the second derivatives~\cite{Vil08}. 

Another question is 
whether there are other conditions that imply eventual log-concavity under the heat flow.
Currently we only know it for when the initial distribution is a convolution of a bounded and a log-concave distribution.
It is interesting to study what happens for a larger class of initial distributions, for example sub-Gaussian.

Alternatively, 
for each point in space we can define the notion of a ``time to log-concavity,'' after which the final distribution is log-concave at that point.
In general, this time is finite for each fixed point, and eventual log-concavity occurs if 
the supremum of this time over space is finite.
There is a generic bound for this time to log-concavity in terms of the variance, and we can prove a slightly better bound under sub-Gaussian assumption, but not much is known.

We are seeking a proof of the nonconvexity of mutual information for the examples presented in~$\S$\ref{Sec:NonConv}.
The nonconvexity is clear from Figure~\ref{Fig:Mixt}, and we have explicit formulae for the second derivatives, but it is desirable to have a formal proof.

It is also interesting to study the effect of dimension in this problem, whether it makes convexity of mutual information easier or more difficult to occur.
From Theorem~\ref{Thm:EventConvFI}, and 
taking into account the growth of Fisher information and fourth moment with dimension, we see that the effect of dimension seems to be to delay the eventual convexity.

Finally, for mixtures of point masses, we have shown that the definition of self-information under the heat flow remembers the discrete initial data.
We can show this also holds for the Ornstein-Uhlenbeck process~\cite{WJOU18}.
It is interesting to study whether the self-information limit is the same under more general flows such as the Fokker-Planck process.

\bibliographystyle{IEEEtran}
\bibliography{mi_arxiv_v2.bbl}

\newpage
\appendix

\section{Proofs}

\subsection{Proof of Lemma~\ref{Lem:DerEnt}}

These identities follow by direct calculation and integration by parts (and Bochner's formula for the second identity).
The first derivative of entropy along the heat flow is De Bruijn's identity~\cite{Sta59}.
The second derivative of entropy along the heat flow is by McKean~\cite{McKean66} in one dimension, and by Toscani~\cite{Tos99} in multi dimension; see also Villani~\cite{Vil00} for a clean proof.

\subsection{Proof of Lemma~\ref{Lem:KJMut}}
\label{App:ProofKJMut}

We first introduce some definitions.
We view the joint distribution $\rho_{XY}(x,y) = \rho_Y(y) \rho_{X|Y}(x\,|\,y)$ as a family of probability distributions $\rho_{X|Y}(\cdot\,|\,y)$ parameterized by $y \in \R^n$, which has distribution $\rho_Y$.
We also assume the density $\rho_{X|Y}(\cdot\,|\,y)$ is smooth with respect to $y$.

The {\em pointwise backward Fisher information matrix} of $X$ given $Y=y$ is
\begin{multline*}
\widetilde \Phi(X\,|\,Y=y) = \\ \int_{\R^n} \rho_{X|Y}(x\,|\,y) (\nabla_y \log \rho_{X|Y}(x\,|\,y))(\nabla_y \log \rho_{X|Y}(x\,|\,y))^\top dx.
\end{multline*}
By integration by parts (assuming boundary terms vanish), we can also write
$$\widetilde \Phi(X\,|\,Y\!=\!y) = -\int_{\R^n} \rho_{X|Y}(x\,|\,y) \nabla^2_y \log \rho_{X|Y}(x\,|\,y) dx.$$
The {\em pointwise backward Fisher information} of $X$ given $Y=y$ is
\begin{align*}
\Phi(X\,|\,Y\!=\!y) &= \Tr(\widetilde \Phi(X\,|\,Y\!=\!y)) \\ 
&= \int_{\R^n} \rho_{X|Y}(x\,|\,y) \|\nabla_y \log \rho_{X|Y}(x\,|\,y)\|^2 dx.
\end{align*}
The {\em backward Fisher information matrix} of $X$ given $Y$ is
$$\widetilde \Phi(X\,|\,Y) = \int_{\R^n} \rho_Y(y) \, \widetilde \Phi(X\,|\,Y\!=\!y) \, dy.$$
The {\em backward Fisher information} of $X$ given $Y$ is
$$\Phi(X\,|\,Y) = \Tr(\widetilde \Phi(X\,|\,Y)).$$
Note $\widetilde \Phi(X\,|\,Y\!=\!y) \succeq 0$ and $\Phi(X\,|\,Y\!=\!y) \ge 0$ for all $y \in \R^n$, so $\widetilde \Phi(X\,|\,Y) \succeq 0$ and $\Phi(X\,|\,Y) \ge 0$.

Similarly, 
the 
{\em pointwise backward second-order Fisher information} of $X$ given $Y=y$ is
$$
\Psi(X\,|\,Y\!=\!y) = 
\int_{\R^n} \rho_{X|Y}(x\,|\,y) \|\nabla^2_y \log \rho_{X|Y}(x\,|\,y)\|^2_{\HS} \, dx.
$$
The {\em backward second-order Fisher information} of $X$ given $Y$ is
$$\Psi(X\,|\,Y) = \int_{\R^n} \rho_Y(y) \, \Psi(X\,|\,Y\!=\!y) \, dy.$$
Note that $\Psi(X\,|\,Y=y) \ge 0$ for all $y \in \R^n$, so $\Psi(X\,|\,Y) \ge 0$.

Finally, the {\em Fisher information matrix} of $Y$ is
$$\widetilde J(Y) = \int_{\R^n} \rho_Y(y) (\nabla_y \log \rho_Y(y))(\nabla_y \log \rho_Y(y))^\top dy.$$
By integration by parts (assuming boundary terms vanish), we can also write
$$\widetilde J(Y) = -\int_{\R^n} \rho_Y(y) \nabla^2_y \log \rho_Y(y) dy.$$
Note that $\widetilde J(Y) \succeq 0$ and Fisher information is its trace: $J(Y) = \Tr(\widetilde J(Y))$.

As stated in~$\S\ref{Sec:Mut}$, mutual Fisher information is in fact equal to the backward Fisher information.

\begin{lemma}\label{Lem:JPhi}
For any joint random variable $(X,Y)$,
$$J(X;Y) = \Phi(X\,|\,Y).$$
\end{lemma}
\begin{proof}
From the factorization
$$\rho_X(x) \rho_{Y|X}(y\,|\,x) = \rho_Y(y) \rho_{X|Y}(x\,|\,y)$$
we have
$$-\nabla^2_y \log \rho_{Y|X}(y\,|\,x) = -\nabla^2_y \log \rho_Y(y) - \nabla^2_y \log \rho_{X|Y}(x\,|\,y).$$
We integrate both sides with respect to $\rho_{XY}(x,y)$.
The left-hand side gives the expected Fisher information matrix $\widetilde J(Y\,|\,X)$.
The first term on the right-hand side gives $\widetilde J(Y)$, while the second term gives the $\widetilde \Phi(X\,|\,Y)$.
That is,
$\widetilde J(Y\,|\,X) = \widetilde J(Y) + \widetilde \Phi(X\,|\,Y)$,
or equivalently,
$$\widetilde J(X;Y) = \widetilde J(Y\,|\,X) - \widetilde J(Y) = \widetilde \Phi(X\,|\,Y).$$
Taking trace gives
$$J(X;Y) = \Tr(\widetilde J(X;Y)) = \Tr(\widetilde \Phi(X\,|\,Y)) = \Phi(X;Y)$$
as desired.
\end{proof}

Similarly, mutual second-order Fisher information can be represented in terms of the backward second-order Fisher information, albeit in a more complicated way.

\begin{lemma}\label{Lem:KPsi}
For any joint random variable $(X,Y)$,
\begin{multline*}
K(X;Y) = \Psi(X\,|\,Y) \, + \\ 2 \int_{\R^n} \rho_Y(y) \langle -\nabla^2 \log \rho_Y(y), \, \widetilde \Phi(X\,|\,Y\!=\!y)\rangle_{\HS} \, dy.
\end{multline*}
\end{lemma}
\begin{proof}
As before we have the decomposition
$$-\nabla^2_y \log \rho_{Y|X}(y\,|\,x) = -\nabla^2_y \log \rho_Y(y) - \nabla^2_y \log \rho_{X|Y}(x\,|\,y).$$
Taking the squared norm on both sides and expanding, we get
\begin{align*}
&\|\nabla^2_y \log \rho_{Y|X}(y\,|\,x)\|^2_{\HS} \\
&~~~~= \|\nabla^2_y \log \rho_Y(y)\|^2_{\HS} + \|\nabla^2_y \log \rho_{X|Y}(x\,|\,y)\|^2_{\HS} \\ 
&~~~~~~~~ + 2 \langle \nabla^2_y \log \rho_Y(y), \nabla^2_y \log \rho_{X|Y}(x\,|\,y) \rangle_{\HS}.
\end{align*}
We integrate both sides with respect to $\rho_{XY}(x,y)$.
On the left-hand side we get $K(Y\,|\,X)$.
The first term on the right-hand side  
gives $K(Y)$; the second term gives $\Psi(X\,|\,Y)$;
for the third term, 
by first integrating over $\rho_{X|Y}(x\,|\,y)$ we obtain an inner product with $\widetilde \Phi(X\,|\,Y\!=\!y)$.
That is,
\begin{multline*}
K(Y\,|\,X) = K(Y) + \Psi(X\,|\,Y) \\
+ 2 \int_{\R^n} \rho_Y(y) \langle -\nabla^2_y \log \rho_Y(y), \widetilde \Phi(X\,|\,Y\!=\!y)\rangle_{\HS} \, dy.
\end{multline*}
This implies the desired expression for $K(X;Y) = K(Y\,|\,X)-K(Y)$.
\end{proof}

We can prove a lower bound for $K(X;Y)$ under log-semiconcavity assumption on $Y$.

\begin{lemma}\label{Lem:KLC}
If $Y \sim \rho_Y$ is $\alpha$-log-semiconcave for some $\alpha \in \R$, then
$$K(X;Y) \ge \Psi(X\,|\,Y) + 2\alpha \Phi(X\,|\,Y).$$
\end{lemma}
\begin{proof}
Since $-\nabla^2 \log \rho_Y(y) \succeq \alpha I$ and $\widetilde \Phi(X\,|\,Y=y) \succeq 0$ for all $y \in \R^n$, we have
\begin{align*}
\langle -\nabla^2 \log \rho_Y(y), \, \widetilde \Phi(X\,|\,Y\!=\!y) \rangle_{\HS}
&\ge \, \langle \alpha I, \, \widetilde \Phi(X\,|\,Y\!=\!y) \rangle_{\HS} \\
&=\, \alpha \Tr(\widetilde \Phi(X\,|\,Y\!=\!y)) \\
&=\, \alpha \, \Phi(X\,|\,Y\!=\!y).
\end{align*}
Integrating with respect to $\rho_Y(y)$ gives
\begin{align*}
&\int_{\R^n} \rho_Y(y) \langle -\nabla^2 \log \rho_Y(y), \, \widetilde \Phi(X\,|\,Y\!=\!y) \rangle_{\HS} \, dy \\
&~~~~~~~ \ge\, \alpha \int_{\R^n} \rho(y) \Phi(X\,|\,Y\!=\!y) \, dy
 \,=\, \alpha \, \Phi(X\,|\,Y).
\end{align*}
Adding $\Psi(X\,|\,Y)$ and using Lemma~\ref{Lem:KPsi} gives the result.
\end{proof}

Furthermore, we have the following result which is reminiscent of the inequality~\eqref{Eq:KJ} between first and second-order Fisher information.

\begin{lemma}\label{Lem:PsiPhi}
For any joint random variable $(X,Y)$ in $\R^n \times \R^n$,
$$\Psi(X\,|\,Y) \ge \frac{\Phi(X\,|\,Y)^2}{n}.$$
\end{lemma}
\begin{proof}
Let $A_{x,y} = -\nabla^2_y \log \rho_{X|Y}(x\,|\,y)$.
By Cauchy-Schwarz inequality,
$$\|A_{x,y}\|^2_{\HS} = \Tr(A_{x,y}^2) \ge \frac{(\Tr(A_{x,y}))^2}{n}.$$
Taking expectation over $(X,Y) \sim \rho_{XY}$ and applying Cauchy-Schwarz again,
we get the desired result
\begin{align*}
\Psi(X\,|\,Y) &= \E[\|A_{X,Y}\|^2_{\HS}] \\
&\ge \frac{\E[(\Tr(A_{X,Y}))^2]}{n} \\
&\ge \frac{(\E[\Tr(A_{X,Y})])^2}{n} \\
&= \frac{\Phi(X\,|\,Y)^2}{n}.
\end{align*}
\end{proof}

Finally, we are ready to prove Lemma~\ref{Lem:KJMut}.

\begin{proof}[Proof of Lemma~\ref{Lem:KJMut}]
By Lemma~\ref{Lem:KLC} and~\ref{Lem:PsiPhi},
$$K(X;Y) \ge \frac{\Phi(X\,|\,Y)^2}{n} + 2\alpha \Phi(X\,|\,Y).$$
Since $J(X;Y) = \Phi(X\,|\,Y)$ by Lemma~\ref{Lem:JPhi}, the result follows.
\end{proof}

\subsection{Proof of Lemma~\ref{Lem:DerMut}}

These identities follow from Lemma~\ref{Lem:DerEnt} and the linearity of the heat flow channel.

Concretely, recall by Lemma~\ref{Lem:DerEnt} that $\frac{d}{dt} H(X_t) = \frac{1}{2} J(X_t)$.
We apply this result to the conditional density $\rho_{X_t|X_0}(\cdot\,|\,x_0)$ to get $\frac{d}{dt} H(X_t\,|\,X_0=x_0) = \frac{1}{2} J(X_t\,|\,X_0=x_0)$ for each $x_0 \in \R^n$.
Taking expectation over $X_0 \sim \rho_0$ and interchanging the order of expectation and time differentiation yields $\frac{d}{dt} H(X_t\,|\,X_0) = \frac{1}{2} J(X_t\,|\,X_0)$.
Combining this with the earlier result above yields
$\frac{d}{dt} I(X_0;X_t) = \frac{1}{2} J(X_0;X_t)$, as desired.
The proof for $\frac{d^2}{dt^2} I(X_0;X_t) = -\frac{1}{2} K(X_0;X_t)$ proceeds identically using the second identity in Lemma~\ref{Lem:DerEnt}.

\subsection{Detail for~$\S\ref{Sec:MutHeat}$}
\label{App:DetMutHeat}

We compute $J(X_t;X_0)$ and $K(X_t;X_0)$ along the heat flow $X_t = X_0 + \sqrt{t} Z$.
Let $X_0 \sim \rho_0$, $X_t \sim \rho_t$, $(X_0,X_t) \sim \rho_{0t}$, and we write the conditionals as
$$\rho_0(x) \rho_{t|0}(y\,|\,x) = \rho_{0t}(x,y) = \rho_t(y) \rho_{0|t}(x\,|\,y).$$
Then
$$-\nabla_x \log \rho_{0|t}(x\,|\,y) = -\nabla_x \log \rho_0(x) - \nabla_x \log \rho_{t|0}(y\,|\,x).$$
Along the heat flow $X_t\,|\,X_0$ is Gaussian with covariance $tI$, so we have explicitly $-\nabla_x \log \rho_{t|0}(y\,|\,x) = \frac{x-y}{t}$.
Therefore,
\begin{align}\label{Eq:DetMutHeatCalc}
-\nabla_x \log \rho_{0|t}(x\,|\,y) = -\nabla_x \log \rho_0(x) +\frac{x-y}{t}.
\end{align}
Take the squared norm on both sides and expand:
\begin{align*}
\|\nabla_x \log \rho_{0|t}(x\,|\,y)\|^2 &= \|\nabla_x \log \rho_0(x)\|^2 + \frac{\|x-y\|^2}{t^2} \\
&~~~~ + \frac{2}{t} \langle -\nabla_x \log \rho_0(x), x-y \rangle.
\end{align*}
Now we take expectation of both sides over $(X_0,X_t)$.
The left-hand side gives $J(X_0\,|\,X_t)$.
The first term on the right-hand side gives $J(X_0)$;
the second term gives $\frac{1}{t^2}\E[\|X_0-X_t\|^2] = \frac{1}{t^2} \E[\|\sqrt{t}Z\|^2] = \frac{n}{t}$ where $Z \sim \N(0,I)$;
while the third term gives $0$ by integrating over $y$ first for each fixed $x$.
That is,
\begin{align}\label{Eq:JXYX}
J(X_0\,|\,X_t) = J(X_0) + \frac{n}{t}.
\end{align}
Therefore,
$$J(X_t;X_0) = J(X_0\,|\,X_t) - J(X_0) = \frac{n}{t}.$$

Next, we differentiate~\eqref{Eq:DetMutHeatCalc} again with respect to $x$ to get
$$-\nabla^2_x \log \rho_{0|t}(x\,|\,y) = -\nabla^2_x \log \rho_0(x) +\frac{I}{t}.$$
Take the squared norm on both sides and expand:
\begin{align*}
&\|\nabla^2_x \log \rho_{0|t}(x\,|\,y)\|_{\HS}^2 \\
&= \|\nabla^2_x \log \rho_0(x)\|_{\HS}^2 + \frac{\|I\|^2_{\HS}}{t^2} 
+ \frac{2}{t} \langle -\nabla^2_x \log \rho_0(x), I \rangle_{\HS} \\
&= \|\nabla^2_x \log \rho_0(x)\|_{\HS}^2 + \frac{n}{t^2} 
- \frac{2}{t} \Delta_x \log \rho_0(x).
\end{align*}
Now we take expectation of both sides over $(X_0,X_t)$.
The left-hand side gives $K(X_0\,|\,X_t)$.
The first term on the right-hand side gives $K(X_0)$;
the second term is a constant;
while the third term gives $\frac{2}{t}J(X_0)$.
That is,
$$K(X_0\,|\,X_t) = K(X_0) + \frac{n}{t^2} + \frac{2}{t} J(X_0).$$
Therefore,
$$K(X_t;X_0) = K(X_0\,|\,X_t) - K(X_0) = \frac{n}{t^2} + \frac{2}{t} J(X_0).$$

\subsection{Proof of Theorem~\ref{Thm:PerpConv}}

Recall that the heat flow preserves log-concavity.
This is because the Gaussian density (the heat kernel) is log-concave, and convolution with a log-concave distribution preserves log-concavity by the Pr\'ekopa-Leindler inequality.

By assumption $X_0 \sim \rho_0$ is log-concave, so $X_t \sim \rho_t$ is also log-concave for all $t \ge 0$.
By Lemma~\ref{Lem:KJMut} and~\ref{Lem:DerMut}, this implies
$\frac{d^2}{dt^2} I(X_0;X_t) = K(X_0;X_t) \ge 0$, 
which means mutual information is always convex.

\subsection{Proof of Theorem~\ref{Thm:EventConv}}
\label{App:ProofEventConv}

Throughout, let $X_t = X_0 + \sqrt{t} Z$ denote the heat flow.
Let $X_0 \sim \rho_0$, $X_t \sim \rho_t$, $(X_0,X_t) \sim \rho_{0t}$, and we write the conditionals as
$$\rho_0(x) \rho_{t|0}(y\,|\,x) = \rho_{0t}(x,y) = \rho_t(y) \rho_{0|t}(x\,|\,y).$$

We first establish the following result to help us determine when we have eventual log-concavity under the heat flow; see Appendix~\ref{App:ProofHesHeat} for the proof.

\begin{lemma}\label{Lem:HesHeat}
Along the heat flow, for all $y \in \R^n$,
\begin{align}\label{Eq:HesHeat}
-\nabla^2_y \log \rho_{t}(y) = \frac{1}{t} \left(I - \frac{1}{t} \Cov(\rho_{0|t}(\cdot\,|\,y))\right).
\end{align}
\end{lemma}

In particular, for bounded initial distribution we have the following eventual log-concavity.

\begin{lemma}\label{Lem:BddLC}
If $X_0 \sim \rho_0$ is $D$-bounded, then along the heat flow, $X_t \sim \rho_t$ is log-concave for all $t \ge D^2$.
\end{lemma}
\begin{proof}
Since $X_0 \sim \rho_0$ is $D$-bounded, the conditional distributions $\rho_{0|t}(\cdot\,|\,y)$ are also $D$-bounded for all $y \in \R^n$.
In particular, 
$\Cov(\rho_{0|t}(\cdot\,|\,y)) \preceq D^2 I.$
Therefore, by Lemma~\ref{Lem:HesHeat},
$$-\nabla^2_y \log \rho_{t}(y) \succeq \frac{1}{t} \left(1 - \frac{D^2}{t}\right) I.$$
If $t \ge D^2$, then $-\nabla^2_y \log \rho_{t}(y) \succeq 0$ for all $y \in \R^n$, which means $X_t \sim \rho_t$ is log-concave.
\end{proof}

We are now ready to prove Theorem~\ref{Thm:EventConv}.

\begin{proof}[Proof of Theorem~\ref{Thm:EventConv}]
By Lemma~\ref{Lem:BddLC}, $X_t \sim \rho_t$ is log-concave for $t \ge D^2$.
By Lemma~\ref{Lem:KJMut} and~\ref{Lem:DerMut}, this implies mutual information is convex for all $t \ge D^2$.
\end{proof}

\subsection{Proof of Corollary~\ref{Cor:EventConv}}

Analogous to Lemma~\ref{Lem:BddLC}, we have the following result.

\begin{lemma}\label{Lem:BddConvLC}
If $X_0 \sim \rho_0$ is a convolution of a $D$-bounded and a log-concave distribution, then along the heat flow, $X_t \sim \rho_t$ is log-concave for all $t \ge D^2$.
\end{lemma}
\begin{proof}
We write $X_0 = B_0+C$ where $B_0$ is a $D$-bounded random variable and $C$ is a log-concave random variable independent of $B$.
Then $X_t = X_0 + \sqrt{t}Z = (B_0 + \sqrt{t}Z) + C = B_t+C$ where $B_t = B_0 + \sqrt{t}Z$ is the heat flow from $B_0$.
By Lemma~\ref{Lem:BddLC}, $B_t$ is log-concave for $t \ge D^2$.
Then by the Pr\'ekopa Leindler inequality, $X_t = B_t+C$ is also log-concave for all $t \ge D^2$.
\end{proof}

We are now ready to prove Corollary~\ref{Cor:EventConv}.

\begin{proof}[Proof of Corollary~\ref{Cor:EventConv}]
By Lemma~\ref{Lem:BddConvLC}, $X_t \sim \rho_t$ is log-concave for $t \ge D^2$.
By Lemma~\ref{Lem:KJMut} and~\ref{Lem:DerMut}, this implies mutual information is convex for all $t \ge D^2$.
\end{proof}

\subsection{Proof of Lemma~\ref{Lem:HesHeat}}
\label{App:ProofHesHeat}

We use the same setting and notation as in Appendix~\ref{App:ProofEventConv}.

We first establish the following result.

\begin{lemma}\label{Lem:HessCov}
Along the heat flow, for all $x,y \in \R^n$,
\begin{align}\label{Eq:HessCov}
-\nabla^2_y \log \rho_{0|t}(x\,|\,y) = \frac{\Cov(\rho_{0|t}(\cdot\,|\,y))}{t^2}.
\end{align}
\end{lemma}
\begin{proof}
We observe that the conditional density $\rho_{0|t}(x\,|\,y)$ can be written as an exponential family distribution over $x$ with parameter $\eta = \frac{y}{t}$:
$$\rho_{0|t}(x\,|\,y) = h(x) e^{\langle x,\eta \rangle - A(\eta)}$$
where $h(x) = \rho_0(x) e^{-\frac{\|x\|^2}{2t}}$ is the base measure, and
$$A(\eta) = \log \int_{\R^n} h(x) e^{\langle x,\eta \rangle} \, dx$$
is the log-partition function, or normalizing constant.
Then we have
$$-\nabla^2_y \log \rho_{0|t}(x\,|\,y) = \frac{1}{t^2} \nabla^2_\eta A(\eta).$$
By a general identity for exponential family~\cite{WJ08}, or simply by differentiating, we have that
$$\nabla^2_\eta A(\eta) = \Cov(\rho_{0|t}(\cdot\,|\,y)).$$
Combining the two expressions above yields the result.
\end{proof}

We are now ready to prove Lemma~\ref{Lem:HesHeat}.

\begin{proof}[Proof of Lemma~\ref{Lem:HesHeat}]
From the factorization
$$\rho_t(y) \rho_{0|t}(x\,|\,y) = \rho_0(x) \rho_{t|0}(y\,|\,x)$$
we have, along the heat flow and by Lemma~\ref{Lem:HessCov},
\begin{align*}
-\nabla^2_y \log \rho_t(y) &= -\nabla^2_y \log \rho_{t|0}(y\,|\,x) + \nabla^2_y \log \rho_{0|t}(x\,|\,y) \\
&= \frac{1}{t} I -\frac{1}{t^2}\Cov(\rho_{0|t}(\cdot\,|\,y)),
\end{align*}
as desired.
\end{proof}

\subsection{Proof of Theorem~\ref{Thm:EventConvFI}}

Let $X_t = X_0 + \sqrt{t} Z$ denote the heat flow.
We first establish some results.

\begin{lemma}\label{Lem:KHeat}
Along the heat flow,
$$K(X_0;X_t) = \frac{2}{t^3} \Var(X_0\,|\,X_t) - \frac{1}{t^4} \E[\|\Cov(\rho_{0|t}(\cdot\,|\,X_t))\|^2_{\HS}].$$
\end{lemma}
\begin{proof}
Squaring and taking the expectation of the identity~\eqref{Eq:HesHeat} in Lemma~\ref{Lem:HesHeat} gives
\begin{align*}
K&(X_t) = \frac{1}{t^2} \E\Big[\Big\|I - \frac{1}{t} \Cov(\rho_{0|t}(\cdot\,|\,X_t))\Big\|^2_{\HS}\Big] \\
&= \frac{\|I\|^2_{\HS}}{t^2} - \frac{2}{t^3} \E[\Var( \rho_{0|t}(\cdot\,|\,X_t) )] \\
&~~~~ + \frac{1}{t^4} \E[\|\Cov(\rho_{0|t}(\cdot\,|\,X_t))\|^2_{\HS}] \\
&= \frac{n}{t^2} - \frac{2}{t^3} \Var(X_0\,|\,X_t) + \frac{1}{t^4}  \E[\|\Cov(\rho_{0|t}(\cdot\,|\,X_t))\|^2_{\HS}].
\end{align*}
Since $K(X_t\,|\,X_0) = n/t^2$, this implies the desired result.
\end{proof}

\begin{lemma}\label{Lem:VarJ}
Assume $J(X_0) < +\infty$.
Along the heat flow,
$$\Var(X_0\,|\,X_t) \ge \frac{n^2}{J(X_0) + \frac{n}{t}}.$$
\end{lemma}
\begin{proof}
For any random variable $X \sim \rho$ in $\R^n$ with a smooth density, recall the uncertainty relationship
$$J(X) \Var(X) \ge n^2,$$
which also follows from the Cauchy-Schwarz inequality and integration by parts.
Applying this to the conditional densities $\rho_{0|t}(\cdot\,|\,y)$ yields
$$\Var(\rho_{0|t}(\cdot\,|\,y)) \ge \frac{n^2}{J(\rho_{0|t}(\cdot\,|\,y))}.$$
Taking expectation over $Y = X_t$ and noting that $\E[\frac{1}{J}] \ge \frac{1}{\E[J]}$ by Cauchy-Schwarz, we get
$$\Var(X_0\,|\,X_t) \ge \E\left[\frac{n^2}{J(\rho_{0|t}(\cdot\,|\,X_t))}\right] \ge \frac{n^2}{J(X_0\,|\,X_t)}.$$
Finally, recall from~\eqref{Eq:JXYX} that $J(X_0\,|\,X_t) = J(X_0)+\frac{n}{t}$.
\end{proof}

Recall that $M_4(X_0)$ is the fourth moment of $X_0$.

\begin{lemma}\label{Lem:CovM4}
Assume $M_4(X_0) < +\infty$.
Along the heat flow,
$$\E[\|\Cov(\rho_{0|t}(\cdot\,|\,X_t))\|^2_{\HS}] \le M_4(X_0).$$
\end{lemma}
\begin{proof}
Let $\mu_0 = \E[X_0]$.
For each $y \in \R^n$,
\begin{align*}
\|\Cov(\rho_{0|t}(\cdot\,|\,y))\|^2_{\HS} 
&\le (\Tr(\Cov(\rho_{0|t}(\cdot\,|\,y))))^2 \\
&= (\Var(\rho_{0|t}(\cdot\,|\,y)))^2  \\
&\le \left( \E_{\rho_{0|t}(\cdot\,|\,y)}[\|X-\mu_0\|^2]\right)^2 \\
&\le \E_{\rho_{0|t}(\cdot\,|\,y)}[\|X-\mu_0\|^4].
\end{align*}
Taking expectation over $Y=X_t$ and applying the tower property of expectation gives the result.
\end{proof}

We are now ready to prove Theorem~\ref{Thm:EventConvFI}.

\begin{proof}[Proof of Theorem~\ref{Thm:EventConvFI}]
By Lemma~\ref{Lem:KHeat},~\ref{Lem:VarJ}, and~\ref{Lem:CovM4}, we have
$$K(X_0;X_t) \ge \frac{2n^2}{t^3(J(X_0)+\frac{n}{t})} - \frac{M_4(X_0)}{t^4}.$$
The right-hand side above is nonnegative if $2n^2t^4 \ge t^3M_4(X_0)(J(X_0)+\frac{n}{t})$, or equivalently, if 
$$2n^2 t^2 - t J(X_0)M_4(X_0) - nM_4(X_0) \ge 0.$$
Therefore, $K(X_0;X_t) \ge 0$ if $t$ is larger than the upper root of the quadratic polynomial above, which is the case when
$$t \ge \frac{J(X_0)M_4(X_0)}{4n^2}\left(1+\sqrt{1+\frac{8n}{J(X_0)^2 M_4(X_0)}}\right).$$
Furthermore, by Cauchy-Schwarz and the uncertainty relationship,
$$J(X_0)^2 M_4(X_0) \ge J(X_0)^2 \Var(X_0)^2 \ge n^4.$$
Plugging this to the bound above and further using $n \ge 1$, we conclude that 
$K(X_0;X_t) \ge 0$, and hence mutual information is convex, whenever
$$t \ge \frac{J(X_0)M_4(X_0)}{n^2}.$$
\end{proof}

\subsection{Proof of Theorem~$\ref{Thm:GenMixt}$}

At each $t > 0$, the density of $X_t \sim \sum_{i=1}^k p_i \N(a_i,tI)$ is
$$\rho_t(y) = \frac{1}{(2\pi t)^{n/2}} \sum_{i=1}^k p_i e^{-\frac{\|y-a_i\|^2}{2t}}.$$
The entropy of $X_t$ is
$$H(X_t)
= \frac{n}{2} \log (2\pi t) - \E\left[ \log \left( \sum_{i=1}^k p_i e^{-\frac{\|X_t-a_i\|^2}{2t}} \right) \right].$$
The expectation is over the mixture $X_t \sim \sum_{i=1}^k p_i \N(a_i, tI)$, which we split into a sum over $i=1,\dots,k$ of the individual expectations over $Y \sim \N(a_i, tI)$.
When $Y \sim \N(a_i, tI)$, we write $Y = a_i + \sqrt{t} Z$ where $Z \sim \N(0,I)$.
Then we can write the entropy above as
{\small
\begin{align*}
&H(X_t) - \frac{n}{2} \log (2\pi t) \\
&= - \sum_{i=1}^k p_i \E\Big[\log\Big(p_i e^{-\frac{\|Z\|^2}{2}} + \sum_{j \neq i} p_j e^{-\frac{\|\sqrt{t}Z+a_i-a_j\|^2}{2t}}\Big)\Big] \\
&= - \sum_{i=1}^k p_i \E\Big[ \log p_i - \frac{\|Z\|^2}{2} + \log \Big(1 +  \sum_{j \neq i} \frac{p_j}{p_i} e^{\frac{\|Z\|^2}{2}-\frac{\|\sqrt{t}Z+a_i-a_j\|^2}{2t}}\Big)\Big] \\
&= h(p) + \frac{n}{2} - \sum_{i=1}^k p_i \E\Big[ \log \Big(1 +  \sum_{j \neq i} \frac{p_j}{p_i} e^{\frac{\|Z\|^2}{2}-\frac{\|\sqrt{t}Z+a_i-a_j\|^2}{2t}}\Big)\Big]
\end{align*}
}
where $h(p) = -\sum_{i=1}^k p_i \log p_i$ is the discrete entropy.

Since $H(X_t\,|\,X_0) = \frac{n}{2} \log (2\pi t e)$,
we have for mutual information
{\small
$$h(p) - I(X_0;X_t) = \sum_{i=1}^k p_i \E\Big[ \log \Big(1 +  \sum_{j \neq i} \frac{p_j}{p_i} e^{\frac{\|Z\|^2}{2}-\frac{\|\sqrt{t}Z+a_i-a_j\|^2}{2t}}\Big)\Big].$$
}

Clearly $h(p)-I(X_0;X_t) \ge 0$ since the logarithm on the right-hand side above is positive.

On the other hand, using the inequality $\log(1+\sum_j x_j) \le \sum_j \log(1+x_j)$ for $x_j > 0$, we also have the upper bound
{\small
$$h(p)-I(X_0;X_t) \le \sum_{i=1}^k p_i \sum_{j \neq i} \E\Big[ \log \Big(1 + \frac{p_j}{p_i} e^{\frac{\|Z\|^2}{2}-\frac{\|\sqrt{t}Z+a_i-a_j\|^2}{2t}}\Big)\Big].$$
}
For each $i \neq j$, the exponent on the right-hand side above is
$$\frac{\|Z\|^2}{2}-\frac{\|\sqrt{t}Z+a_i-a_j\|^2}{2t} = -\frac{\langle Z,a_i-a_j\rangle}{\sqrt{t}} - \frac{\|a_i-a_j\|^2}{2t},$$
which has the $\N(-\frac{\|a_i-a_j\|^2}{2t},\frac{\|a_i-a_j\|^2}{t})$ distribution in $\R$, so it has the same distribution as $-\frac{\|a_i-a_j\|^2}{2t} + \frac{\|a_i-a_j\|}{\sqrt{t}} Z_1$ where $Z_1 \sim \N(0,1)$ is the standard one-dimensional Gaussian.
Thus, we can write the upper bound above as
$$h(p) - I(X;Y) \le \sum_{i=1}^k p_i \sum_{j \neq i} \E\left[ \log \left(1 + b_{ij} e^{c_{ij} Z_1 - \frac{c_{ij}^2}{2}}\right)\right]$$
where $b_{ij} = \frac{p_j}{p_i}$ and $c_{ij} = \frac{\|a_i-a_j\|}{\sqrt{t}}$, and $Z_1 \sim \N(0,1)$ in $\R$.

By Lemma~\ref{Lem:Log2} below, if $c_{ij} \ge \max\{1,\frac{26}{b_{ij}}\}$, then we have
\begin{align*}
h(p) - I(X_0;X_t) \le 3 \sum_{i=1}^k p_i \sum_{j \neq i} b_{ij} e^{-0.085c_{ij}^2}.
\end{align*}
Note that $b_{ij} = \frac{p_j}{p_i} \le \|p\|_\infty$ and $c_{ij}^2 = \frac{\|a_j-a_i\|^2}{t} \ge \frac{m^2}{t}$, so
\begin{align*}
h(p) - I(X_0;X_t) &\le 3 \sum_{i=1}^k p_i \sum_{j \neq i} \|p\|_\infty e^{-0.085 \frac{m^2}{t}} \\
&= 3(k-1) \|p\|_\infty e^{-0.085 \frac{m^2}{t}}.
\end{align*}
Now, the condition $c_{ij} \ge \max\{1,\frac{26}{b_{ij}}\}$ is equivalent to 
$$t \le \frac{\|a_i-a_j\|^2}{\max\{1,\frac{26}{b_{ij}}\}^2}.$$
Since $\|a_i-a_j\|^2 \ge m^2$ and $\frac{1}{b_{ij}} = \frac{p_i}{p_j} \le \|p\|_\infty$,
 the condition above is satisfied when
$$t \le \frac{m^2}{\max\{1,26\|p\|_\infty\}^2} = \frac{m^2}{26^2\|p\|_\infty^2}.$$

Thus, we conclude that if $t \le \frac{m^2}{676\|p\|_\infty^2}$, then
\begin{align*}
h(p) - I(X_0;X_t) \le 3 (k-1) \|p\|_\infty e^{-0.085\frac{m^2}{t}}
\end{align*}
as desired.

To complete the proof of Theorem~\ref{Thm:GenMixt}, it remains to prove the following estimate.

\begin{lemma}\label{Lem:Log2}
Let $b > 0$, $c \ge \max\{1,\frac{26}{b}\}$, and $Z \sim \N(0,1)$.
Then
$$\E[\log(1 + be^{cZ-\frac{c^2}{2}})\big] \le 3be^{-0.085 c^2}.$$
\end{lemma}
\begin{proof}
We use the standard tail bound $\Pr(Z \ge x) \le \frac{1}{\sqrt{2\pi}} \frac{1}{x} e^{-\frac{x^2}{2}}$ for all $x > 0$, which follows from using the inequality $1 \le \frac{z}{x}$ in the integration.
In particular, for $x \ge \frac{1}{\sqrt{2\pi}}$ we have $\Pr(Z \ge x) \le e^{-\frac{x^2}{2}}$.

Let $0 < \eta < 1$.
We split the expectation into three parts:
\begin{enumerate}
  \item For $Z < (1-\eta)\frac{c}{2}$: We have $cZ-\frac{c^2}{2} < -\eta\frac{c^2}{2}$, so $\log(1 + be^{cZ-\frac{c^2}{2}}) \le \log(1+be^{-\eta\frac{c^2}{2}}) \le be^{-\eta\frac{c^2}{2}}$.
 The contribution to the expectation from this region is at most $be^{-\eta\frac{c^2}{2}} \Pr(Z < (1-\eta)\frac{c}{2}) \le be^{-\eta\frac{c^2}{2}}$.
 
 \item For $(1-\eta)\frac{c}{2} \le Z < \frac{c}{2}$: We have $cZ-\frac{c^2}{2} < 0$, so $\log(1 + be^{cZ-\frac{c^2}{2}}) \le \log(1+b) \le b$.
 The contribution to the expectation from this region is at most $b \Pr((1-\eta)\frac{c}{2} \le Z < \frac{c}{2}) \le b \Pr(Z \ge (1-\eta)\frac{c}{2}) \le b e^{-(1-\eta)^2\frac{c^2}{8}}$ where the last inequality holds for $c \ge \frac{2}{(1-\eta)\sqrt{2\pi}}$.
 
 \item For $Z \ge \frac{c}{2}$: We have $cZ-\frac{c^2}{2} \ge 0$, so $\log(1 + be^{cZ-\frac{c^2}{2}}) \le \log((1+b)e^{cZ-\frac{c^2}{2}}) = \log(1+b) + cZ - \frac{c^2}{2} \le b + cZ$.
 The contribution to the expectation from this region is at most
 $\int_{\frac{c}{2}}^\infty (b+cz) \frac{1}{\sqrt{2\pi}} e^{-\frac{z^2}{2}} dz
 = b \Pr(Z \ge \frac{c}{2}) + \frac{c}{\sqrt{2\pi}} e^{-\frac{c^2}{8}}
 \le (b+c) e^{-\frac{c^2}{8}}$,
 where the last inequality holds for $c \ge \frac{2}{\sqrt{2\pi}}$.
\end{enumerate}

Combining the three parts above, we have that for $c \ge \frac{2}{1-\eta}$,
$$\E[\log(1 + be^{cZ-\frac{c^2}{2}})] \le be^{-\eta\frac{c^2}{2}} + b e^{-(1-\eta)^2\frac{c^2}{8}} + (b+c) e^{-\frac{c^2}{8}}.$$
The leading exponent is $\min\{\eta, \frac{(1-\eta)^2}{4}\} \frac{c^2}{2}$, which is maximized by $\eta^\ast = 3-\sqrt{8} \approx 0.1716$.
Set $\eta = \eta^\ast$.
Note that for $c \ge \frac{2}{(\frac{1}{4}-\eta^\ast)b}$ we have $\frac{c^2}{2}(\frac{1}{4}-\eta^\ast) \ge \frac{c}{b} \ge \log(1+\frac{c}{b})$, so $(b+c)e^{-\frac{c^2}{8}} \le b e^{-\eta^\ast \frac{c^2}{2}}$.

Thus, for $c \ge \max\{\frac{2}{(1-\eta^\ast)\sqrt{2\pi}}, \frac{2}{(\frac{1}{4}-\eta^\ast)b}\}$, we have 
$$\E\left[\log\left(1 + be^{cZ-\frac{c^2}{2}}\right)\right] \le 3be^{-\eta^\ast\frac{c^2}{2}}.$$
Since $\frac{\eta^\ast}{2} \approx 0.0858 > 0.085$, $\frac{2}{(1-\eta^\ast)\sqrt{2\pi}} \approx 0.9631 < 1$, and $\frac{2}{\frac{1}{4}-\eta^\ast} \approx 25.5014 < 26$, we can simplify this conclusion by saying that for $c \ge \max\{1,\frac{26}{b}\}$ we have 
$$\E\left[\log\left(1 + be^{cZ-\frac{c^2}{2}}\right)\right] \le 3be^{-0.085 c^2},$$
as desired.
\end{proof}

\subsection{Proof of Corollary~\ref{Cor:Last}}

From Theorem~\ref{Thm:GenMixt}, we have for small $t$,
$$\left|\frac{I(X_0;X_t) - h(p)}{t^\ell} \right| \le 3(k-1) \|p\|_\infty \frac{e^{-0.085 \frac{m^2}{t}}}{t^\ell}.$$
Inductively, this implies all derivatives of $I(X_0;X_t)$ tend to $0$ exponentially fast as $t \to 0$.

\end{document}